\documentclass[11pt]{article}   	                 
\usepackage{amsmath,amsthm, amssymb,amsfonts,amscd,mathabx}

\usepackage[colorlinks=true,linkcolor=blue,citecolor=red]{hyperref}
\usepackage[pdftex]{graphicx}                 
\usepackage{wrapfig}                               
\usepackage[font=footnotesize]{caption}                     
\usepackage{float}
\usepackage{subcaption}
\usepackage[defaultcolor=red]{changes}

\numberwithin{equation}{section}
\newtheorem{theorem}{Theorem}[section]

\newtheorem{proposition}[theorem]{Proposition}

{\begin{trivlist}\item[]\textbf{Proof#1 }}%
{\qed\end{trivlist}}

 {\begin{trivlist}\item[]\textbf{Acknowledgments }}{\end{trivlist}}

\usepackage[english]{babel}
\graphicspath{ {.},{./images} }

\newcommand{\R}{\ensuremath{\mathbb{R}}}

\newcommand{\Z}{\ensuremath{\mathbb{Z}}}
\newcommand{\N}{\ensuremath{\mathbb{N}}}

\renewcommand{\epsilon}{\varepsilon}
\newcommand{\mc}{\mathcal}

\newcommand{\rmi}{\mathrm{i}}
\newcommand{\rmd}{\mathrm{d}}
\newcommand{\rmO}{\mathrm{O}}

\title{}
\author{}
\allowdisplaybreaks

\begin{document}

\begin{center}
{\fontsize{15}{15}\fontfamily{cmr}\fontseries{b}\selectfont{Anchored spirals in the driven curvature flow approximation
}}\\[0.2in]
Nan Li and Arnd Scheel\footnote{The authors acknowledge partial support by the National Science Foundation through grant  NSF DMS-2205663. } \\
\textit{\footnotesize University of Minnesota, School of Mathematics,   206 Church St. S.E., Minneapolis, MN 55455, USA}\\[3mm]
\emph{In memory of Claudia Wulff}
\end{center}

\begin{abstract}\noindent 
We study existence, asymptotics, and stability of spiral waves in a driven curvature approximation, supplemented with an anchoring condition on a circle of finite radius.  We analyze the motion of curves written as graphs in polar coordinates, finding spiral waves as rigidly rotating shapes. The existence analysis reduces to a planar ODE and asymptotics are given through center manifold expansions. In the limit of a large core, we find rotation frequencies and corrections starting form a problem without curvature corrections using geometric singular perturbation theory. Finally, we demonstrate orbital stability of spiral waves by exploiting a comparison principle inherent to curvature driven flow. 
\end{abstract}

\section{Introduction and main results}

Spiral waves are both fascinating and ubiquitous in excitable and oscillatory media. They represent topological defects in the phase of spatio-temporal oscillations and act as pacemakers, sending out periodic wave trains with an intrinsic, selected frequency; see\cite{MR4580296}, \cite{MR0025140}, \cite{MR0572965}, \cite{Belmonte1997}, \cite{desai_kapral_2009}, and references therein.
%
%
Existence of spiral waves is known only in some special cases. In oscillatory media, existence was shown in a limiting regime of balanced linear and nonlinear dispersion, first in $\lambda-\omega$-systems \cite{MR0588502}, \cite{MR0639765}, \cite{MR0665385}, \cite{MR0359550}, \cite{MR0639764}
and then more generally near a Hopf bifurcation in reaction-diffusion systems  \cite{MR1638058}.
In the case of excitable media, there do not appear to be rigorous results on the existence of spiral waves; see however \cite{KEENER1986307}, \cite{TYSON1988327}, \cite{BERNOFF1991125}
for asymptotic results and \cite{Barkley_cookbook}
for selection recipes. 

On the other hand, it was recognized early on that wave fronts are well described by purely geometric evolution equations, at least in a limit of small curvature. At leading order, one describes the evolution of a spiral arm locally as motion in a normal direction to the front with speed given by $V+D\kappa$, where $V$ is the speed of propagation of a planar interface, $\kappa$ a signed curvature, and $D>0$ a line tension parameter. The description of spiral waves by a single line segment is incorrect near the center of rotation, in the core region, where the line segment terminates. Alternative descriptions using two line segments, one for the wave front and one for the wave back run into similar difficulties in this core region, where front and back merge \cite{KEENER1986307,BERNOFF1991125}. 

We focus here on \emph{anchored} spiral waves. In this situation, an inhomogeneity or a hole in the domain break the translation symmetry and the spiral filament attaches to the boundary of this core region. In the simplest scenario, the filament is attached to a circle with a fixed contact angle and rotates locally around the circle. One observes globally a filament that curls up and converges to a spiral wave with a selected frequency. 

The existence of spiral waves in such a scenario appears to be well recognized in the physics literature \cite{MIKHAILOV1991379,MR1257848,MR4656059,YAMADA1993153}, and there are a number of more recent existence proofs in the mathematical literature, also for more general relations between curvature and normal velocity \cite{MR2341216,MR2094384,MR2187763,MR1629103}. These results are particularly concerned with situations when the spiral tip, rather than rotating on a fixed circle, follows more intricate trajectories such as epicycloids. This intriguing phenomenon, spiral meander and drift, was first related to Euclidean symmetries by Dwight Barkley in \cite{PhysRevLett.72.164}. A first rigorous analysis in Claudia Wulff's doctoral thesis \cite{wulffthesis,wulffthesispub} identified the difficulties associated with a rigorous bifurcation analysis in a corotating frame due to apparent infinite velocities at infinity. Both this  and some of her later work as well as work from other groups \cite{fssw,ssw2,ssw3,ssw1,GLM} analyze the complexity induced by instability or heterogeneity in the medium. Left open in those works is the relevance (or rather the apparent irrelevance)  of the presence of  essential spectrum on the imaginary axis: all reductions rely crucially on spectral gap assumptions, which fail for Archimedean spirals; see also \cite{MR4580296} for a review and discussion of spectral properties. 

Our work complements these results in several aspects. First, we analyze the existence problem in polar coordinates, which gives us more direct access to the shape of spirals in physical space, with a more direct relation to the partial differential equations for which the curvature flow arises as a singular limit.  We also analyze the large-core limit, when the size of the hole to which the spiral attaches is large, and find expansions for the frequency of the spiral. Lastly, our coordinate choice allows us to study stability of spiral waves. The literature is surprisingly scarce on stability results for spiral waves, with conceptual considerations in  \cite{MR4580296};
see also references therein for instability mechanisms. Partial stability results based on matched asymptotics calculations are available in the  case of $\lambda-\omega$-systems \cite{MR0665385}
with however little insight into nonlinear stability. In that respect the result here can be thought of as the first nonlinear stability result for spiral waves, although we recognize that the simplification of a curvature approximation and anchoring simplify the stability problems significantly without obvious generalizations to reaction-diffusion systems. 

To state our main results, we consider a curve $\gamma(t,s)$ in $\R^2\setminus \{|x|< R\}$ parameterized by arclength $s\geq0$ and time $t$ which evolves with velocity $V+D\kappa$ in the normal direction $N=-\gamma_{ss}/|\gamma_{ss}|$. Here, $\kappa=\langle\gamma_{ss},N\rangle$ is the signed curvature, $V>0$ the velocity of a straight line segment, and $D>0$ a coefficient modeling line tension. We assume that the curve is anchored with $|\gamma(t,0)|=R$ and $\gamma_s(t,0)$ perpendicular to the circle $\{|x|=R\}$.
%
\begin{figure}[ht]
    \centering
    \includegraphics[width=.3\textwidth]{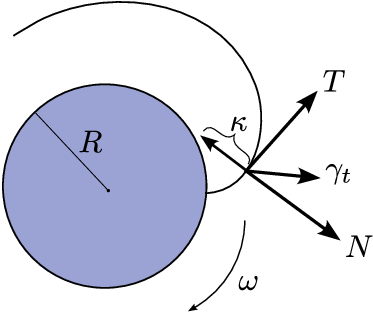}
    \caption{Schematic picture of spiral filament anchored at a disc of radius $R$, evolving pointwise with velocity $\gamma_t$, with effective normal velocity given by $V+\kappa D$. The normal motion can lead to an effective rigid rotation with angular frequency $\omega$.}
    \label{fig:enter-label}
\end{figure}
We say $\gamma(t,s)$ is rigidly rotating with frequency $\omega$ if 
\[
\gamma(t,s)=R_{\omega t}\gamma(0,s),\qquad \text{where } R_\varphi =\left(\begin{array}{cc} \cos\varphi& \sin\varphi\\ -\sin\varphi &\cos\varphi \end{array}\right),
\]
that is, $\omega>0$ corresponds to clockwise rotation; see Figure \ref{fig:enter-label}.
Moreover, we refer to $\gamma_0(s)=\gamma(0,s)$ as an asymptotically Archimedean spiral if its intersection with rays is asymptotically linear,
\[
R_{\varphi} \gamma_0(s)\cap \R_+\times\{0\}= \{(r_j(\varphi),0),j\in\N\}, \quad \varphi\in[0,2\pi)],
\]
with $R\leq r_0(\varphi)<r_1(\varphi)<\ldots$, with $r_j'(\varphi)>0$, and with
\[
\lim_{j\to\infty} r_{j+1}(\varphi)-r_j(\varphi) = 2\pi/\ell,\quad \text{ for some }\ell>0.
\]
The inverse distance $\ell$, which is independent of $\varphi$, is an asymptotic wavenumber. We say the Archimedean spiral is rotating outward if $\omega>0$. 
\begin{theorem}[Existence]\label{theorem existence}
    For all $D,V,R>0$, there exists an outward rotating asymptotically Archimedean spiral $\gamma_0$ with frequency $\omega=\omega_\mathrm{sp}(D,V)>0$.
\end{theorem}
We restate this result in polar coordinates and give a proof in Section \ref{s:3}.

\begin{theorem}[Asymptotics --- large core]\label{theorem large core}
    For fixed $D,V>0$, and $R\gg 1$, we have the expansion
    \[\omega_\mathrm{sp} = VR^{-1} - \sigma_0 2^{1/3} D^{2/3} V^{1/3} R^{-5/3} + \rmO(R^{-7/3}), \qquad  \sigma_0=1.01879297\ldots\]
    Moreover, the shape of the spiral is given explicitly through a zero-curvature approximation, formally setting $D=0$, at leading order.
\end{theorem}
A more detailed expansion is formulated in Section \ref{s:4}, Theorem \ref{theorem existence in polar coordidate}, also with an explicit expression for the zeroth order approximation \eqref{e:wr}. While the first-order term simply reflects the travel time of the line segment around the circle of radius $R$, the correction term, which reflects a slow-down of motion due to curvature effects, is less intuitive and, to our knowledge, not documented in the literature. 
\begin{theorem}[Stability --- Informal]\label{theorem stability}
    All curves $\gamma(s)$ that are sufficiently close to $\gamma_0(s)$ constructed in Theorem \ref{theorem existence} will stay close to $\gamma_0$ for all times.
\end{theorem}
We state a precise version  and outline  a proof in Section \ref{s:5}. To prepare proofs in Sections \ref{s:3}--\ref{s:5}, we reformulate the curve evolution as a quasilinear parabolic equation for curves written as graphs in polar coordinates, next.

\section{Rotating spirals in polar coordinates}
\label{s:2}

We consider plane curves written as graphs in polar coordinates, with covering space $\R$ for the angle,
\[
\Gamma=\{\gamma(t,r)=(r\cos(\Phi(t,r)),r\sin(\Phi(t,r))\mid\,r\geq R,\,t\geq 0\},
\]
for some smooth function $\Phi:[0,\infty)\times [R,\infty)\to \R$. In order to derive the evolution equation for $\Phi$ from the normal velocity $V+D\kappa$, we first compute unit 
tangent vector and arclength $s$ from
\begin{align}
    T &= \frac{\gamma_r}{|\gamma_r|} = \frac{1}{\sqrt{1+r^2\Phi_r^2}} \left(\cos(\Phi)-r\Phi_r\sin(\Phi),\sin(\Phi)+r\Phi_r\cos(\Phi)\right),\qquad
    \frac{\rmd r}{\rmd s} &=  \frac{1}{\sqrt{1+r^2\Phi_r^2}},
\end{align}
with scalar product and induced norm from Cartesian coordinates.
We orient the normal such that it points ``downwards'' at $r=R$, when $\Phi=0,\Phi_r=0$,
\begin{align}
    N &=  \frac{1}{\sqrt{1+r^2\Phi_r^2}} (\sin(\Phi)+r\Phi_r\cos(\Phi),-\cos(\Phi)+r\Phi_r\sin(\Phi)).
\end{align}
From this, we find the signed curvature
\[\kappa = \left\langle\frac{\rmd T}{\rmd s},N\right\rangle=-\frac{r\Phi_{rr}+r^2\Phi_r^3+2\Phi_r}{(1+r^2\Phi_r^2)^{3/2}}
,\]
using the Euclidean scalar product; see Figure \ref{fig:enter-label}, where in particular $\Phi_r=0$ and $\Phi_{rr}>0$ at $r=R$, so that $\kappa<0$ in the orientation given by $N$. We find the normal velocity of the curve by projecting
$
    \gamma_t = (-r\Phi_t \sin(\Phi),r\Phi_t \cos(\Phi)),
$
onto $N$ along $T$, and equating with  $V+D\kappa$, which gives
\begin{equation}\label{pde}
    \Phi_t = \frac{Dr\Phi_{rr}-V(1+r^2\Phi_r^2)^{3/2}+Dr^2\Phi_r^3+2D\Phi_r}{r(1+r^2\Phi_r^2)}.
\end{equation}
The curve $\gamma$ is outward rotating where $\Phi_t\cdot\Phi_r<0$. One may choose to scale space $r$ and time $t$ as
\begin{equation}\label{e:scal}
t=\tilde{t}D/V^2,\quad r= \tilde{r}{D/V}, \quad \tilde{\Phi}(\tilde{t},\tilde{r})=\Phi(t,r),
\end{equation}
to find the flow with normalized velocity and line tension $V=D=1$,
\begin{equation}\label{pde nondim}
    \tilde{\Phi}_{\tilde{t}} = \frac{\tilde{r}\tilde{\Phi}_{\tilde{r}\tilde{r}}-(1+\tilde{r}^2\tilde{\Phi}_{\tilde{r}}^2)^{3/2}+\tilde{r}^2\tilde{\Phi}_{\tilde{r}}^3+2\tilde{\Phi}_{\tilde{r}}}{\tilde{r}(1+\tilde{r}^2\tilde{\Phi}_{\tilde{r}}^2)}.
\end{equation}
in the rescaled region $\tilde{r}>\tilde{R}=R/\sqrt{D/V}$. We state results below for the unscaled version,  but will carry out proofs in the simple case $V=D=1$. 

Rigid rotation with angular velocity $\omega$ translates simply into $\Phi(t,r)=\phi(r)-\omega t$, which gives the nonlinear, second-order, non-autonomous ordinary differential in $\phi$
\begin{equation}\label{ode}
    -\omega = \frac{Dr\phi_{rr}-V(1+r^2\phi_r^2)^{3/2}+Dr^2\phi_r^3+2D\phi_r}{r(1+r^2\phi_r^2)}.
\end{equation}
Again, scaling $D=V=1$ yieldds the scaled frequency $\omega=\tilde{\omega}V^2/D$.
We can write (\ref{ode}) as a system of first order autonomous ODEs, setting $\phi_r=\ell$, and appending $\alpha=1/r$ as a dependent variable, 
\begin{align}\label{vf}
    &\begin{cases}
        \ell' = -\frac{\omega}{D} (\alpha^2 + \ell^2) + \frac{V}{D} (\alpha^2 + \ell^2)^{3/2} - 2 \alpha^3 \ell - \alpha \ell^3\\
        \alpha' = -\alpha^4
    \end{cases}
    \qquad \left(\,'\, = \frac{d}{d\tau}\right),\qquad \tau=(r^3-R^3)/3.
\end{align}
The solution $\phi(r)$ is then recovered from 
\[
\phi'(r)=\ell(\tau(r)), \qquad \phi(r)=\int_R^r \phi'(\rho)\rmd \rho.
\]
Note that the introduction of $\alpha$ as an independent variable compactifies the phase space, including $\alpha=0$ which corresponds to $r=\infty$. The new time $\tau$ regularizes the equation, removing singularities of the form $1/\alpha^2$ in the vector field that arises in ``time'' $r$. 

An asymptotically Archimedean spiral in this context is represented by a solution of (\ref{vf}) that tends to an equilibrium point on the $\ell$-axis. Outward rotation, for $-\Phi_t=\omega>0$ then corresponds to $\ell>0$. The solution satisfies the boundary condition $\phi'(R)=0$ if the trajectory originates on the $\ell=0$-axis.

The next two sections will be concerned with an analysis of this system \eqref{vf}.




\section{Existence of rigidly rotating spirals}
\label{s:3}
We analyze \eqref{vf} with boundary condition corresponding to the anchored core at $r=R$, that is, $\alpha_*=1/R$. The perpendicular contact angle is encoded in $\phi_r=0$, that is, $\ell_*=0$.

We will show that $\phi_r=\ell$ tends to an equilibrium point, that is, $\lim_{\tau \to \infty} \ell(\tau) = \omega/V>0$, which indicates that the spiral is Archimedean and outward rotating in the farfield. Moreover, the core radius $R$ uniquely determines  the angular velocity $\omega$, that is, for every $\alpha_*$ there is a unique $\omega_*$ that admits a solution with the above properties. We restate Theorem \ref{theorem existence} in terms of the equation (\ref{vf}).

\begin{theorem}\label{theorem large core in polar coordinate}
    Fix $D,V>0$ and let $(\ell(\tau;\omega),\alpha(\tau;\omega))$ denote the solution of (\ref{vf}) with initial condition $(\ell(0),\alpha(0)) = (\alpha_*,0)$ and parameter $\omega$. We then have that for every $\alpha_* > 0$ there exists a unique $\omega_*>0$ such that
    \[
    \lim_{\tau \to \infty} \ell(\tau;\omega_*) = \frac{\omega_*}{V}
    .\]
    Moreover, $\omega_*$ is strictly increasing in $\alpha_*$.
\end{theorem}
\begin{proof}
    We find solutions to \eqref{vf} shooting from the boundary condition $\{(\ell,\alpha)|\ell=0\}$ to the asymptotic equilibrium $\ell=\omega/V$, $\alpha=0$. It turns out that the asymptotic  equilibrium possesses a one-dimensional center manifold and is otherwise unstable. Solutions in the center-manifold converge to the equilibrium, all other initial conditions diverge. Exploiting in addition a monotonicity in $\omega$, this will show both existence and uniqueness. We start by constructing suitable invariant regions and then find the desired orbit using a somewhat standard inf/sup-construction. 
    
    Fix $\alpha_* > 0$, and assume $D=V=1$. Define
    \begin{align*}
        A_\omega &= \left\{(\alpha,\ell) : \ell > 2\omega\left(1-\frac{\alpha}{4\omega}\right), \alpha > 0, \ell > 0\right\},\\
        B_\omega &= \left\{(\alpha,\ell) : \alpha^2 + \ell^2 < \omega^2, \alpha > 0\right\},
    \end{align*}
    and, writing $\Psi_\tau(\alpha_*,\ell_*;\omega)$ for the solution $(\ell(\tau),\alpha(\tau))$ to \eqref{vf}  with initial condition $(\alpha_*,\ell_*)$ and parameter value $\omega_*$,
    \begin{align*}
        S &= \left\{\omega : \exists\, \tau \in (0,\infty)\;\text{ s.th. } \Psi_\tau(\alpha_*,0;\omega) \in A_\omega\right\},\\
        L &= \left\{\omega : \exists\, \tau \in (0,\infty)\;\text{ s.th. } \Psi_\tau(\alpha_*,0;\omega) \in  B_\omega\right\}.
    \end{align*}
    We note the following properties.
    \begin{enumerate}
        \item The slope of the vector field along the $\alpha$-axis is positive when $\alpha_* < \omega$ and  negative when $\alpha_* > \omega$, since
        \[\frac{d\ell}{d\alpha}\biggr\rvert_{\ell = 0} = -\frac{1}{\alpha^2}(\alpha - \omega).\]
        As a consequence, for any  $\alpha_* < \omega$  there exists a $\epsilon>0$ such that $\ell(\tau;\omega)<0$ for  $\tau \in (0,\epsilon]$. Similarly, for any $\alpha_* < \omega$   there exists a $\epsilon>0$ such that $\ell(\tau;\omega)>0$ for  $\tau \in (0,\epsilon]$.
        \item The $\ell$-axis is invariant since $\alpha'=0$ when $\alpha=0$.
        \item If $0 < \omega \ll 1$, then $A_\omega$ is forward invariant since for $\alpha>0$,
        \[\lim_{\omega \to 0} \,\,\frac{\rmd\ell}{\rmd\alpha}\biggr\rvert_{\ell = 2\omega\left(1-\frac{\alpha}{4\omega}\right)} = -\frac{5^{3/2} }{8\alpha}-\frac{9}{8} < -\frac{1}{2},\]
        that is, the slope of the vector field along the line segment $\ell = 2\omega\left(1-\frac{\alpha}{4\omega}\right)$ is less than the slope $-1/2$  of the line segment when $\omega$ is sufficiently small. Together with (1) and (2), the claim follows immediately.
        \item $B_\omega$ is forward invariant because the direction of the vector field points inward along the perimeter of the half circle $\{(\alpha,\ell) : \alpha^2+\ell^2 = \omega, \alpha > 0\}$, that is,
        \[\frac{\rmd}{\rmd\tau}(\alpha^2+\ell^2)\biggr\rvert_{\ell = \pm \sqrt{\omega^2 - \alpha^2}} = -\omega^4\alpha < 0.\]
        \item $S$ and $L$ are nonempty because $(\alpha_*,\infty) \subset S$ and $(0,\alpha_*/4) \subset L$.
        \item $S$ and $L$ are open because $\Psi_{\tau}(\alpha_*,0;\omega)$ depends continuously on $\omega$.
        \item $S \cap L = \emptyset$ since $A_\omega$ and $B_\omega$ are forward invariant.
        \item $S$ is bounded above: Take $\omega > \alpha$ and note that forward invariance of $B_\omega$ and (1) imply $\Psi_{\tau}(\alpha_*,0;\omega) \notin A_\omega$ for all $\tau>0$.
        \item $L$ is bounded below: Take $\omega < \alpha/2$ and note that forward invariance of $A_\omega$ and (1) imply $\Psi_{\tau}(\alpha_*,0;\omega) \notin B_\omega$ for all $\tau>0$.
    \end{enumerate}
    Thus, $\sup S<\infty$ and $\inf L>0$ exist. Set $\omega_* = \inf L$ and note that $\omega_*$ has the following properties.
    \begin{enumerate}
        \item[(i)] Since $S$ and $L$ are open and disjoint, $\omega_* \notin S \cup L $ and hence $\Psi_\tau(\alpha_*,0;\omega_*) \notin A_\omega \cup B_\omega$ for all $\tau > 0$.
        \item[(ii)] The trajectory  $\{\Psi_\tau(\alpha_*,0;\omega_*),\tau>0\}$ is contained in the bounded open region $\{(\ell,\alpha>0\}\setminus \overline{(A_{\omega_*}\cup B_{\omega_*})}=:G_{\omega_*}$ and hence global in $\tau$.
    \end{enumerate}
    \begin{figure}[ht]
        \centering
        \begin{subfigure}{0.4\textwidth}
            \includegraphics[width=.85\textwidth]{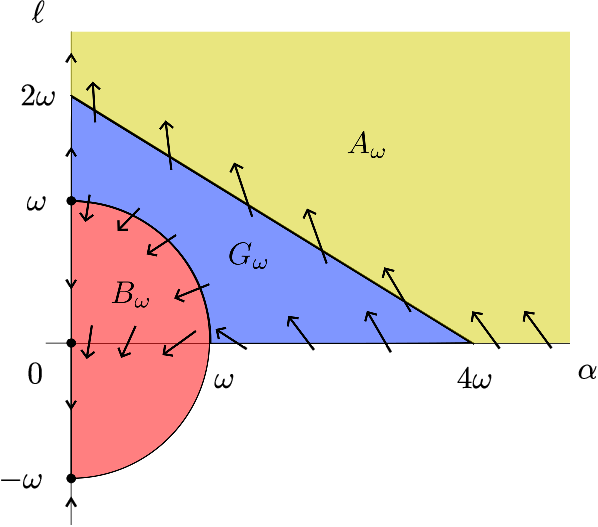}
            \caption{}
        \end{subfigure}
        \hspace{0.5in}
        \begin{subfigure}{0.4\textwidth}
            \includegraphics[width=.85\textwidth]{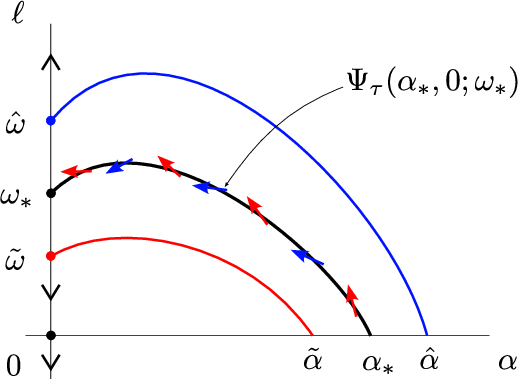}
            \caption{}
        \end{subfigure}
        \caption{(a) Invariant regions $A_\omega$ and $B_\omega$ on the phase plane. The trajectory of $\Psi_{\tau}(\alpha_*,0;\omega)$ stays in the region $G_\omega$ for $\tau > 0$. (b) The $\ell$-coordinate of equilibrium point $\omega_*$ is strictly increasing in $\alpha_*$; the region below the trajectory at $\omega_*$ is forward invariant for $\hat{\omega}>\omega_*$ (indicated by blue arrows)  showing that the associated $\hat{\alpha}$ necessarily satisfies $\hat{\alpha}>\alpha_*$; similarly the  region above the trajectory at $\omega_*$ is forward invariant for $\tilde{\omega}<\omega_*$ (indicated by red arrows)  showing that the associated $\tilde{\alpha}$ necessarily satisfies $\tilde{\alpha}<\alpha_*$.}
    \end{figure}    
    By monotonicity of $\alpha$, the limit set is therefore contained in $\overline{G_{\omega_*}}\cap\{\alpha=0\}$. Inspecting the flow on $\alpha=0$, and noting invariance of the limit set, we conclude that $\Psi_\tau(\alpha_*,0;\omega_*)\to (0,\omega_*)$ for $\tau\to\infty$. 
    

    It remains to show that $\omega_*$ is increasing in $\alpha_*$. Therefore, consider $\tilde{\alpha}<\alpha_*$. Since the center manifold is exponentially repelling, there is a unique trajectory that converges to the equilibrium  $(0,\omega_*)$ on the $\ell$-axis. Assume $\tilde{\omega}>\omega_*$. But then the region below the trajectory of $\{\Psi_\tau(\alpha_*,0;\omega_*),\tau\geq0\}$ is forward invariant for the flow with $\omega=\tilde{\omega}$, since $\omega$ appears in the equation for $\ell'$, only, with a definite negative sign. The initial condition $(\tilde{\alpha},0)$ is contained in this region, but the equilibrium $(0,\tilde{\omega})$ is not, showing that $\tilde{\omega}<\omega_*$.

    Similarly, given $\alpha_*$ and two values $\omega_1<\omega_2$ for which there exists a connection, we find that the trajectory originating at $(\alpha_*,0)$ for $\omega_2$ is contained in the region below the trajectory for $\omega_1$ and hence cannot connect to $(0,\omega_2)$ which lies above this trajectory, a contradiction. This shows strict monotonicity of $\alpha_*$ as a function of $\omega_*$ and hence establishes the uniqueness claim. 
\end{proof}
The following proposition collects some farfield asymptotics.
\begin{proposition}\label{p:inf}
    The solution $(\ell,\alpha)(\tau)$ with parameter $\omega=\omega_*$ and initial condition $\omega=\omega_*$ can be written as $\ell=\lambda(\alpha)$, where $\lambda$ possesses the (not necessarily convergent) expansion at the origin
    \[
    \lambda(\alpha)=\frac{\omega_*}{V} + \sum_{j=0}^\infty \lambda_j\alpha^j,
    \]
    with 
    \[
    \lambda_1=\frac{D \omega_* }{V^2},\quad \lambda_2=\frac{2 D^2 \omega_* ^2-V^4}{2 V^3 \omega_* },\quad \text{and }\ \lambda_3=\frac{D \left(D^2 \omega_* ^2+V^4\right)}{V^4 \omega_* }.
    \]
\end{proposition}
\begin{proof}
    The trajectory is the backward extension of the center manifold at the equilibrium $(0,\omega_*/V,0)$ and the expansion can be readily computed recursively requiring invariance at any order. 
        {
    We therefore substitute the ansatz
    \[
    \ell = \frac{\omega}{V} + \lambda_1 \alpha + \lambda_2 \alpha^2 + \lambda_2 \alpha^3+\rmO(\alpha^4)
    \]
    into
    \[
    \ell' = \frac{d\ell}{d\alpha} \alpha' ,
    \]
    and substitute the Taylor expansion of the expression for $\ell'$ and $\alpha'$ in \eqref{vf} near $\ell=\omega/V$ and $\alpha=0$.  The left-hand side becomes $\ell'=-\lambda_1 \alpha^4 - \lambda_2 \alpha^5 - \lambda_3 \alpha^6$ which indicates that the coefficients of $\alpha, \alpha^2$, and $\alpha^3$ in the Taylor series of the right-hand side necessarily vanish. Solving these three equations successively for $\lambda_1,\lambda_2$, and $\lambda_3$ gives the expressions as stated.
    }
\end{proof}
We remark that the coefficient $\lambda_1$ agrees with the leading-order coefficient computed in a much more general scenario, \cite[(3.7)]{MR4580296},  if one substitutes accordingly 
\[
d_\perp=D, \quad c_\mathrm{g}=V, \quad k_*=\omega_*/V,
\]
for the effective transverse diffusivity $d_\perp$, the group velocity of wave trains emitted by the spiral  $c_\mathrm{g}$, and the asymptotic wavenumber $k_*$. The linear convergence of the wavenumber $\ell$ in $\alpha$ is of course equivalent to a convergence with $1/r$, which translates into a logarithmic divergence of the phase, identified also in  \cite[(3.7)]{MR4580296}.
It would be interesting to determine if higher-order coefficients can also be inferred from a, possibly more refined, approximation by geometric evolution equations. 

We conclude this section with some remarks on other boundary-value problems. First, problem in annuli can be studied with only minor modifications. We are then interested in curves that are anchored at an inner circle of radius $R_\mathrm{i}$ and an outer circle of radius $R_\mathrm{o}>R_\mathrm{i}$, that is, $\ell(R_\mathrm{i})=\ell(R_\mathrm{o})=0$. Our shooting approach immediately generalizes to this setup.
\begin{proposition}\label{p:finite}
Given $R_\mathrm{o}>R_\mathrm{i}>0$, there exists a unique frequency $\omega=\Omega(R_\mathrm{o},R_\mathrm{i})$ and a unique profile $\ell(r)$ with $\ell(R_\mathrm{i})=\ell(R_\mathrm{o})=0$, $\ell(r)>0$ on $(R_\mathrm{i},R_\mathrm{o})$. Moreover, fixing $R_\mathrm{i}$, $\omega(\cdot,R_\mathrm{i})$ is strictly monotonically decreasing with limits $\omega(R_\mathrm{i},R_\mathrm{i})=V/R_\mathrm{i}$, and $\omega(\infty,R_\mathrm{i})=\omega_*$ with $\omega_*$ from Theorem \ref{theorem large core in polar coordinate}.
\end{proposition}
Lastly, we emphasize that our approach gives uniqueness only in the class of curves that can be written as smooth graphs $\phi(r)$. We suspect that introducing a pseudo-angle $\arctan(\ell)$, one may be able to identify solutions that are no graphs in the radial variable, thus possess poles of $\ell=\phi'$.

\section{Expansion of frequency in the large-core limit}
\label{s:4}

We derive the asymptotic expansion of the spiral wave solution for large effective core radius $\tilde{R}=VR/D \to \infty$.
\begin{theorem}\label{theorem existence in polar coordidate}
    Given $\alpha_*>0$, let $\omega_*$ and $\ell=\lambda(\alpha)$ be the solution from Theorem \ref{theorem large core in polar coordinate}. We then have the expansions
    \begin{align}
        \omega_* &= V\alpha_* - \sigma_0 \sqrt[3]{2D^2 V} \alpha_*^{5/3} + \rmO(\alpha_*^{7/3}),\nonumber\\
        \lambda(\alpha) &= \sqrt{\frac{\omega_*^2}{V^2} - \alpha^2} + \rmO({\omega\alpha}), \quad \text{for } \alpha<(1-\delta)\frac{\omega_*}{V} \ \text{ and some } \delta>0,
    \end{align}
     where $\sigma_0=1.01879297\ldots$ is determined by the first zero of the derivative of the Airy function, that is, $\mathrm{Ai}'(-\sigma_0)=0$, $\mathrm{Ai}'(-\sigma)>0$ for $\sigma<\sigma_0$.
\end{theorem}
In order to understand the singular limit, we set $D=V=1$ and study the limit $\omega,\alpha\to 0$, by rescaling (\ref{vf}) as
\begin{align}\label{vf omega}
    \begin{cases}
        \ell_1 = \frac{\ell}{\omega}\\
        \alpha_1 = \frac{\alpha}{\omega}\\
        \tau_1 = \omega^2 \tau
    \end{cases}
    \Longrightarrow
    &\begin{cases}
        \frac{d \ell_1}{d\tau_1} = -(\alpha_1^2 + \ell_1^2) + \frac{V}{D}(\alpha_1^2 + \ell_1^2)^{3/2} - \omega(2\alpha_1^3 \ell_1 + \alpha_1 \ell_1^3)\\
        \frac{d \alpha_1}{d\tau_1} = -\omega \alpha_1^4\\
        \frac{d \omega}{d\tau_1} = 0.
    \end{cases}
\end{align}
Formally setting $\omega=0$, \eqref{vf omega} possesses  a circle of equilibria $\sqrt{\alpha_1^2+\ell_1^2}=1$. In order to prove Theorem \ref{theorem existence in polar coordidate}, we need to understand the dynamics near this circle for small $\omega>0$.  
We therefore rely on geometric singular perturbation theory, as introduced in Fenichel's seminal paper \cite{MR524817} together with methods from an extension to nonhyperbolic points from  \cite{MR1857972}.

\begin{proof}[\textbf{Proof of Theorem \ref{theorem existence in polar coordidate}}]
{Set $V=D=1$.} Consider equation (\ref{vf omega}). For $\omega = 0$,  the circle $\alpha_1^2 + \ell_1^2 = 1$ clearly is an invariant manifold. Linearizing at an equilibrium $(\ell_1^*,\alpha_1^*)$ on this manifold, we find
\[\begin{pmatrix}
    \frac{d \ell_1}{d\tau_1}\\
    \frac{d \alpha_1}{d\tau_1}
\end{pmatrix} = 
\begin{pmatrix}
    \ell_1^* & \alpha_1^*\\
    0 & 0\\
\end{pmatrix}
\begin{pmatrix}
    \ell_1\\
    \alpha_1
\end{pmatrix}\]
with eigenvalues $\ell_1^*$ and $0$. The manifold is normally hyperbolic in the sense of \cite{MR524817} whenever $\ell_1^*\neq 0$, that is, whenever the zero eigenvalue associated with the tangent space of the manifold is the only eigenvalue on the imaginary axis. In particular, the point $\ell_1^*=0,\alpha_1^*=1$ is the only non normally hyperbolic point. Any compact portion, $\{\ell_1=\sqrt{1-\alpha_1^2},0\leq \alpha_1\leq 1-\delta\}$, $\delta>0$ fixed, then satisfies the assumptions of \cite{MR524817} and persists as a $C^k$ invariant manifold, depending in a $C^k$ fashion on $\omega$ for any finite $k<\infty$. We refer to this manifold as $\mathcal{M}_\omega$ and note that we can write it as a graph 
\[\ell_1=\psi(\alpha_1;\omega)=\sqrt{1 - \alpha_1^2}+\rmO(\alpha_1\omega)
,\]
noticing that the correction terms vanish at $\alpha_1=0$. Next, note that the flow  induced on $\mathcal{M}_\omega$  is monotone in $\alpha$, given by the differential equation in $\alpha_1$ when $\mathcal{M}_\omega$ is written as a graph over the $\alpha$-axis. The shape of the spiral in this regime is of course simply given by 
\[
\ell(r)=\omega\psi(\frac{1}{\omega r};\omega)=\sqrt{\omega^2-\frac{1}{r^2}}+\rmO(\frac{\omega}{r}).
\]
We next wish to track this manifold backward in time $\tau$ until it intersects the $\ell_1=0$-axis. It turns out that a neighborhood of the point $(\alpha_1,\ell_1)=(1,0)$ for $\omega\sim 0$ has been analyzed in  \cite{MR1857972}.  In fact, for $\omega=0$, $\alpha_1$ is constant and can be thought of as a parameter. From this point of view, the manifold of equilibria exhibits a saddle-node bifurcation at the critical parameter-value, precisely the problem analyzed in  \cite{MR1857972}, where a typical bifurcation delay of order $\omega^{2/3}$ was identified together with the exact leading-order precoefficient; see Figure \ref{f:passagefold} for an illustration.  We therefore first perform a sequence of coordinate changes that puts our equation in the normal form studied there. 
\begin{figure}
    \centering
    \begin{subfigure}{0.4\textwidth}
        \includegraphics[width=.85\textwidth]{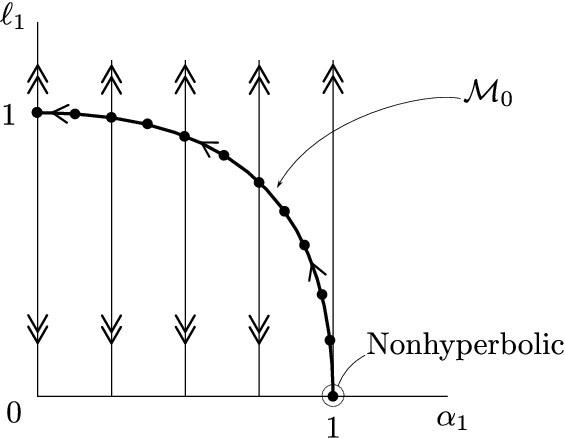}
        \caption{}
    \end{subfigure}
    \hspace{0.3in}
    \begin{subfigure}{0.35\textwidth}
        \includegraphics[width=.85\textwidth]{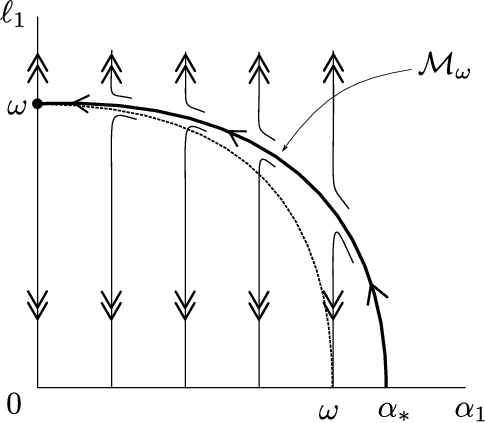}
        \caption{}
    \end{subfigure}
    \caption{Figure (a): The phase portrait of (\ref{vf omega}) with $\omega=0, V=D=1$ and its invariant manifold $\mc M_0$. The point $(\ell_1,\alpha_1) = (0,1)$ is not normally hyperbolic. Figure (b): The invariant manifold $\mathcal{M}_\omega$ where $\omega > 0$.}
    \label{f:passagefold}
\end{figure}

First set $\tilde{\alpha}_1 = \alpha_1 - 1$ to translate the fold point to the origin,
\[\begin{cases}
    \frac{d \ell_1}{d\tau_1} = -\frac{1}{D}((\tilde{\alpha}_1 + \frac{1}{V})^2 + \ell_1^2) + \frac{V}{D}((\tilde{\alpha}_1 + \frac{1}{V})^2 + \ell_1^2)^{3/2} - \omega(2(\tilde{\alpha}_1 + \frac{1}{V})^3 \ell_1 + (\tilde{\alpha}_1 + \frac{1}{V}) \ell_1^3)\\
    \frac{d\tilde{\alpha}_1}{d\tau_1} = -\omega (\tilde{\alpha}_1 + \frac{1}{V})^4\\
    \frac{d\omega}{d\tau_1} = 0.
\end{cases}\]
then expand and rescale as follows,
\begin{align*}
    &\begin{cases}
        \frac{d \ell_1}{d\tau_1} = \tilde{\alpha}_1 + \frac{1}{2} \ell_1^2 + \rmO(\omega,\tilde{\alpha}_1 \ell_1,\ell_1^2,\tilde{\alpha}_1^3)\\
        \frac{d\tilde{\alpha}_1}{d\tau_1} = \omega(-1 + \rmO(\omega,\ell_1,\tilde{\alpha}_1))\\
        \frac{d\omega}{d\tau_1} = 0.
    \end{cases}\\
    \begin{cases}
        \tau_2 = -\tau_1\\
        \ell_2 = -\ell_1\\
        \alpha_2 = -\tilde{\alpha}_1
    \end{cases}
    \Longrightarrow
    &\begin{cases}
        \frac{d \ell_2}{d\tau_2} = -\alpha_2 + \frac{1}{2}\ell_2^2 + \rmO(\omega,\alpha_2 \ell_2,\ell_2^2,\alpha_2^3)\\
        \frac{d \alpha_2}{d\tau_2} = \omega(-1 + \rmO(\omega,\ell_2,\alpha_2))\\
        \frac{d \omega}{d\tau_2} = 0
    \end{cases}\\
    \begin{cases}
        \ell_3 = \lambda \ell_2\\
        \alpha_3 = \mu \alpha_2\\
        \tau_3 = \delta \tau_2
    \end{cases}
    \Longrightarrow
    &\begin{cases}
        \frac{d \ell_3}{d\tau_3} = -\frac{\lambda}{\delta \mu} \alpha_3 + \frac{1}{2 \delta \lambda}\ell_3^2 + \rmO(\omega,\alpha_3 \ell_3,\ell_3^2,\alpha_3^3)\\
        \frac{d \alpha_3}{d\tau_3} = \omega(-\frac{\mu}{\delta} + \rmO(\omega,\ell_3,\alpha_3))\\
        \frac{d \omega}{d\tau_3} = 0.
    \end{cases}
\end{align*}
Choose
\[\lambda = 2^{-2/3},\quad \mu = 2^{-1/3},\quad \delta = 2^{-1/3},
\]
so that
\[\frac{\lambda}{\delta \mu} = \frac{1}{2 \delta \lambda} = \frac{\mu}{\delta} = 1,\]
and obtain
\[\begin{cases}
    \frac{d \ell_3}{d\tau_3} = -\alpha_3 + \ell_3^2 + \rmO(\omega,\alpha_3 \ell_3,\ell_3^2,\alpha_3^3)\\
    \frac{d \alpha_3}{d\tau_3} = -\omega+ \omega \rmO(\omega,\ell_3,\alpha_3))\\
    \frac{d \omega}{d\tau_3} = 0,
\end{cases}
\]
the ``normal form'' for the slow passage studied in \cite{MR1857972}. 

\begin{figure}[ht]
\centering
\begin{subfigure}{0.4\textwidth}
    \includegraphics[width=\textwidth]{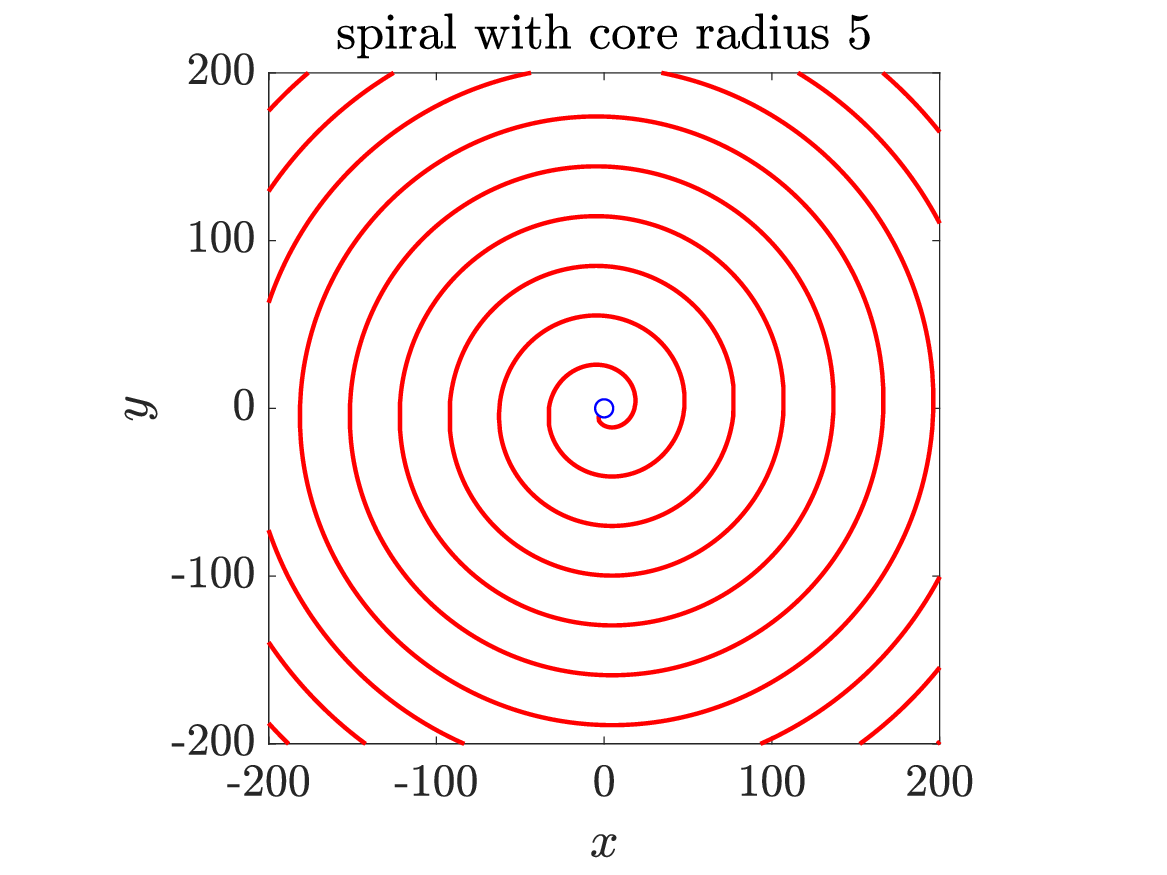}
    \caption{}
\end{subfigure}
\begin{subfigure}{0.4\textwidth}
    \includegraphics[width=\textwidth]{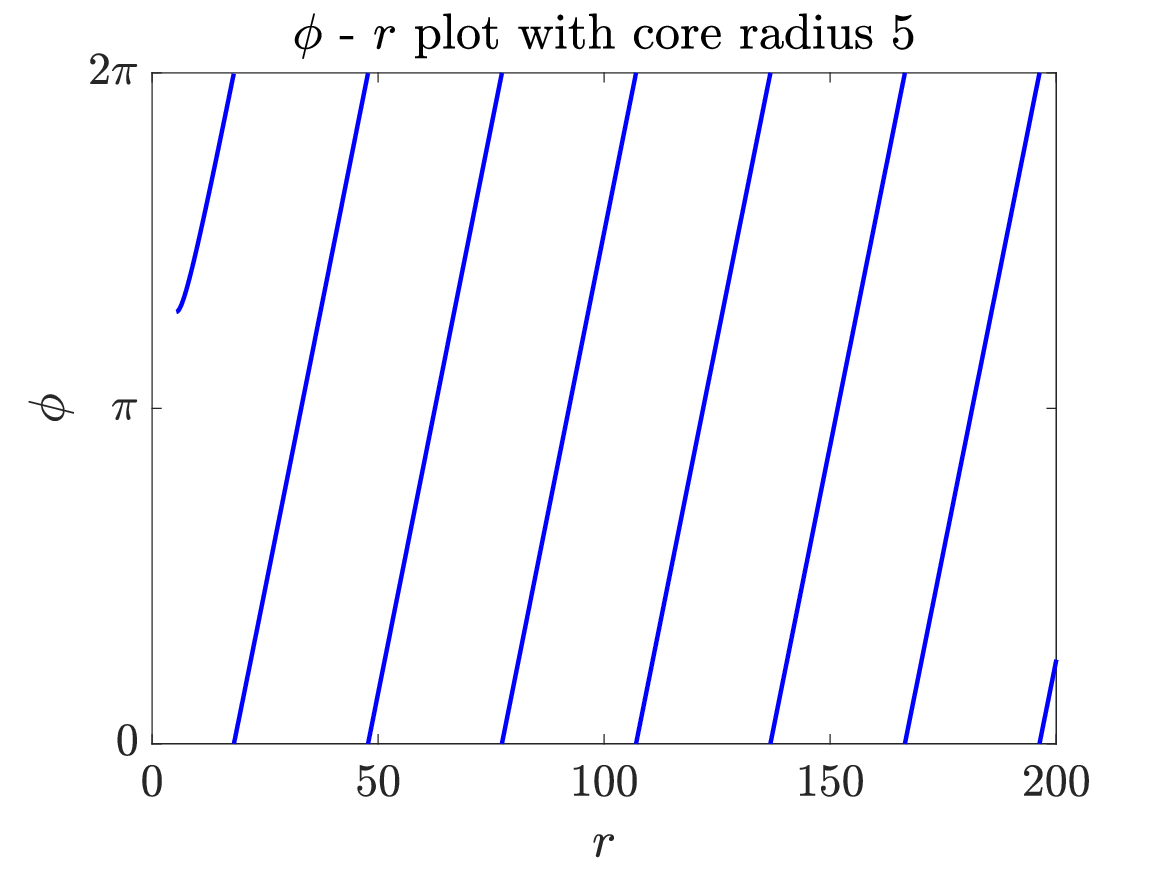}
    \caption{}
\end{subfigure}
\hfill
\begin{subfigure}{0.4\textwidth}
    \includegraphics[width=\textwidth]{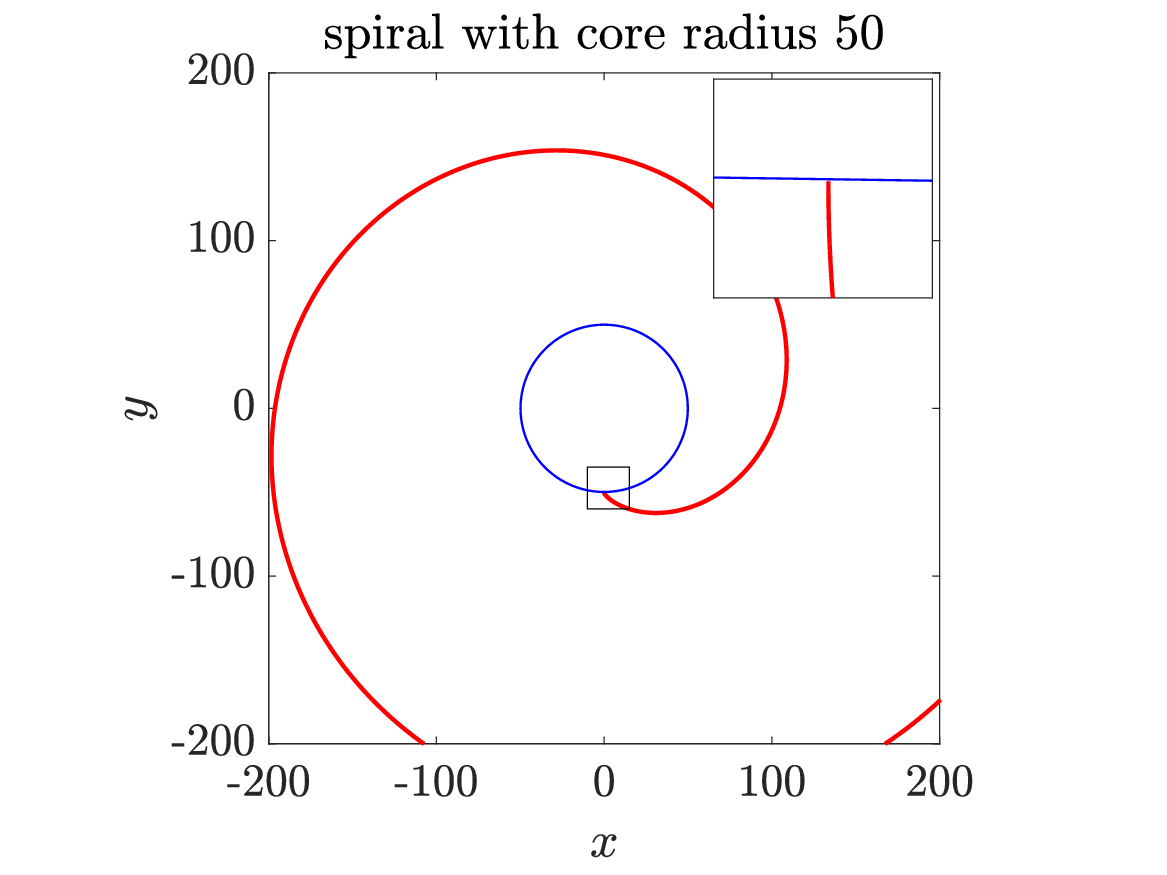}
    \caption{}
\end{subfigure}
\begin{subfigure}{0.4\textwidth}
    \includegraphics[width=\textwidth]{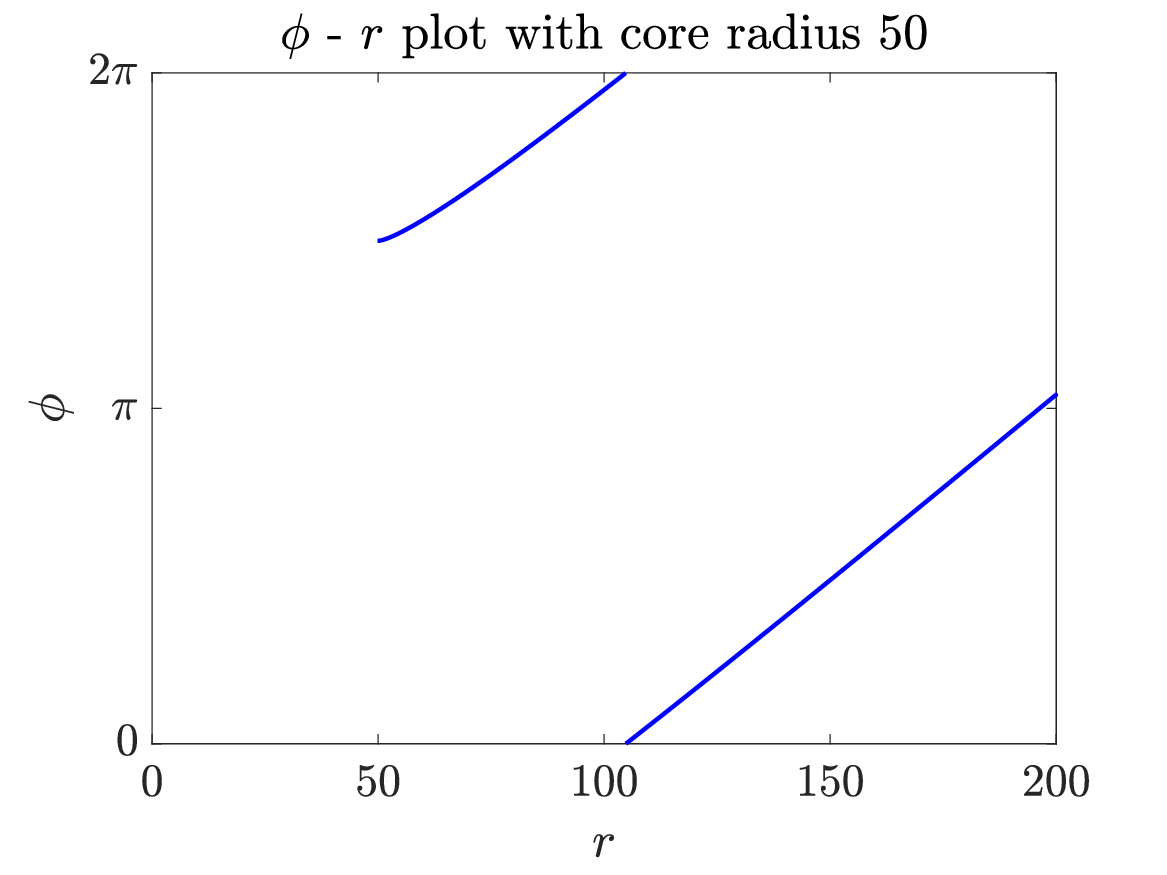}
    \caption{}
\end{subfigure}
        
\caption{Graphs of spirals with various core radii in the plane (left) and in $r-\phi$-plots; for all graphs, $V = D = 1$; (a)-(b): $R = 5$, (c)-(d): $R=50$.}
\end{figure}

\begin{figure}
\centering
\begin{subfigure}{0.4\textwidth}
    \includegraphics[width=\textwidth]{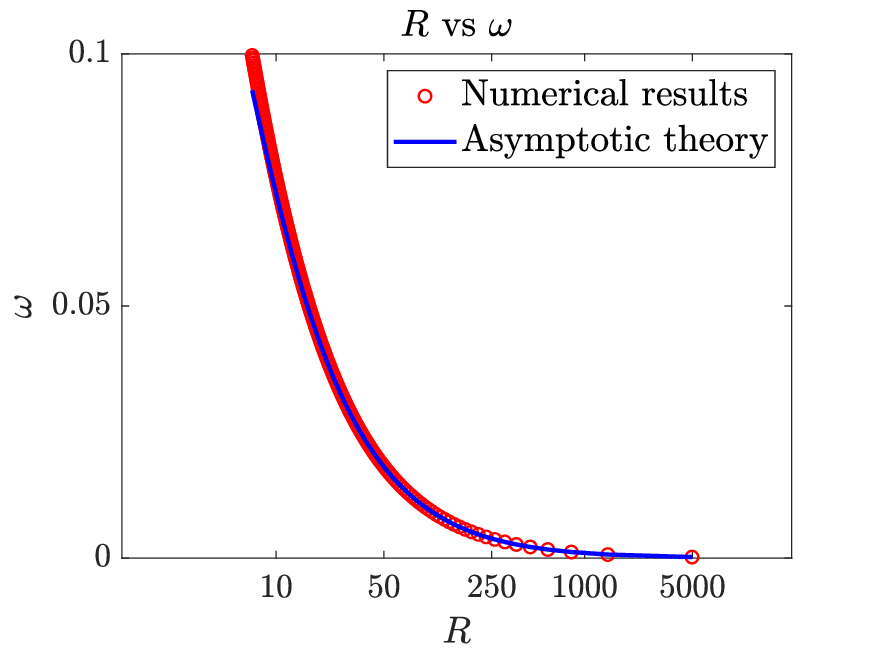}
    \caption{}
\end{subfigure}
\begin{subfigure}{0.4\textwidth}
    \includegraphics[width=\textwidth]{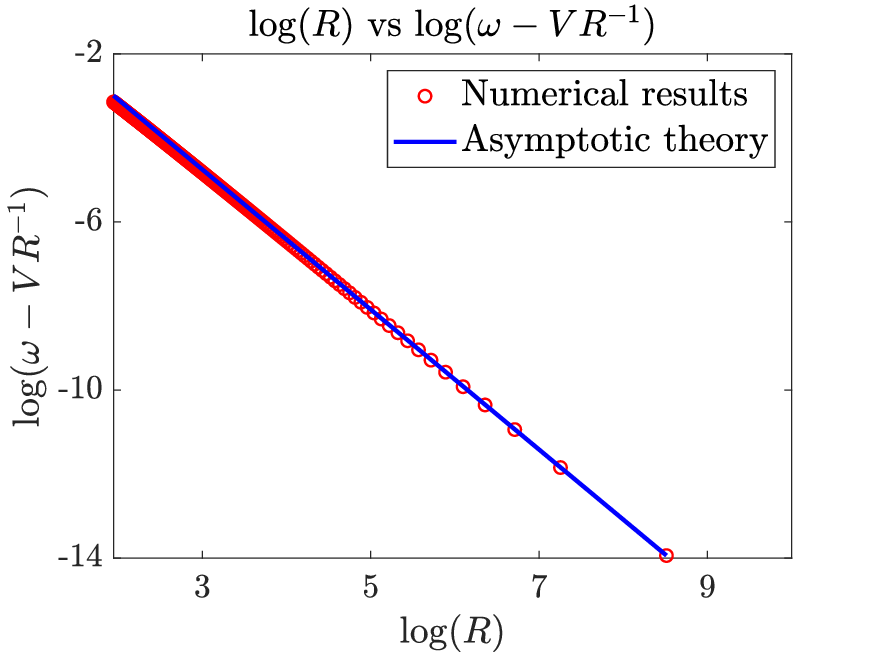}
    \caption{}
\end{subfigure}

\caption{Comparison of the theoretical prediction from Theorem \ref{theorem existence in polar coordidate} and computational results with $V=D=1$. Figure (a): Theoretical expansion $\omega=VR^{-1}-\sigma_0 2^{1/3} D^{2/3} V^{1/3} R^{-5/3}$ (blue solid line) and the numerical results (red dotted line). This shows $\omega = R^{-1}$ asymptotically. Figure (b): Theoretical expansion $\omega-VR^{-1}=-\sigma_0 2^{1/3} D^{2/3} V^{1/3} R^{-5/3}$ (blue solid line) and the numerical results (red dotted line). The $\log$-$\log$ plots coincide on a line of slope $-5/3$ which validates the expansion $\omega = V R^{-1} + \mathcal{O}(R^{-5/3})$. We confirmed that the difference between numerical results and the theoretical prediction in the right-hand plot is of order $R^{-7/3}$.}
\end{figure}

At this point the result needed here differs slightly from the results for the slow passsage, where one studies the exit of the trajectory from a small neighborhood of the origin, that is, the first time a solution hits the section $\ell_3=\rho$ for some $\omega$-independent constant $\rho>0$. This location is found at leading order from the rescaled equation, where $\ell_3=\omega^{1/3}\tilde{\ell}_3$, $\alpha_3=\omega^{2/3}\tilde{\alpha}_3$, $'=\omega^{1/3}\ \tau_3=\omega^{-1/3}\tilde{\tau}_3 $,
\[\begin{cases}
    \frac{d \tilde{\ell}_3}{d\tilde{\tau}_3} = -\tilde{\alpha}_3 + \tilde{\ell}_3^2 + \rmO(\omega^{1/3})\\
    \frac{d \tilde{\alpha}_3}{d\tilde{\tau}_3} = -1+ \rmO(\omega^{1/3})\\
    \frac{d \omega}{d\tau_3} = 0.
\end{cases}
\]
Neglecting the $\rmO(\omega^{1/3})$-terms, we find the Riccati equation  $\frac{d \tilde{\ell}_3}{d\tilde{\tau}_3}=\tilde{\tau}_3+\tilde{\ell}_3^2$, with a unique solution that satisfies the required backward asymptotics $\tilde{\ell}_3\sim \sqrt{-\tilde{\tau}_3}$, given explicitly through 
\[
\tilde{\ell_3}(-\tilde{\tau}_3)=\mathrm{Ai}'(\tilde{\tau}_3)/\mathrm{Ai(\tilde{\tau}_3)},
\]
and $\mathrm{Ai}(\sigma)$ is the unique bounded solution to $u_{\sigma\sigma}-\sigma u=0$. The first zero of $\mathrm{Ai}(\sigma)$ hence determines the blowup time, $\mathrm{Ai}(\sigma_1)=0$, $\tilde{\tau}_3=-\sigma_1$. The geometric desingularization analysis in \cite{MR1857972} shows that this gives indeed yields the leading order contribution to the passage time, that is $\ell(\tau_3)=-\rho$ when $\alpha_3(\tau_3)=-\sigma_1\omega^{2/3}+\rmO(\omega\log(\omega))$. 

In this leading approximation, we find that $\ell(\tau_3)=-\rho$ when at leading order $\alpha_3(\tau_3)=-\sigma_0\omega^{2/3}$, where $-\sigma_0$ is the largest zero of the derivative of the Airy function, $\mathrm{Ai}'(-\sigma_0)=0$, $\mathrm{Ai}(-\sigma)>0$ for $\sigma>\sigma_0$, and   $\sigma_0=1.01879297\ldots$. Inspecting the proof in \cite{MR1857972}, one finds that error terms are determined readily from the passage in the $K_1$-chart, in their notation, with error terms 
\begin{equation}
     \alpha_3(\tau_3)=-\sigma_0\omega^{2/3}+\rmO(\omega).
\end{equation}
Going back through the coordinate changes, we find after a short calculation that the locally invariant manifold $\mc M_\omega$ intersects the line $\ell_1=0$ at $(\ell_1,\alpha_1) = (0,\alpha_*)$ where
\[\alpha_* = \omega + 2^{1/3} \sigma_0 \omega^{5/3} + \rmO(\omega^2).\]
Solving for  $\omega$ in terms of $\alpha_*$, we find, as claimed,
\[
\omega = \alpha_* - 2^{1/3} \sigma_0 \alpha_*^{5/3} + \rmO(\alpha_*^{7/3})
.\]
Using the scaling expressions \eqref{e:scal} for $R_*=1/\alpha_*$ and $\omega$, we now quickly find the expansion as stated. 
\end{proof}
Note that the leading order approximation is given through $\ell=\sqrt{\omega^2/V^2-1/r^2}$. which upon integration gives the Wiener-Rosenblueth spiral,
\begin{equation}\label{e:wr}
  \Phi(r)=  \sqrt{ \frac{\omega^2r^2}{V^2}-1} - \arctan{\sqrt{ \frac{\omega^2r^2}{V^2}-1}}; 
\end{equation}
see also \cite[(3.126)]{MR4656059}\footnote{Note that there is a typo in the sign of the $\arctan$ contribution, there; the associated plots \cite[Fig 3.16]{MR4656059}, there, are for the correct sign that we show here.} and \cite{MR0025140},
\section{Stability}
\label{s:5}
We consider \eqref{pde}, in a corotating frame, 
\begin{equation}\label{e:pde}
    \Phi_t = \frac{Dr\Phi_{rr}-V(1+r^2\Phi_r^2)^{3/2}+Dr^2\Phi_r^3+2D\Phi_r}{r(1+r^2\Phi_r^2)}+\omega_*,\quad r>R,\qquad \Phi_r(R)=0,
\end{equation}
with initial condition $\Phi(r,0)=\phi_*(r)+\varphi(r)$, and with $\phi_*'(r)=\lambda(1/r)$, and   $\omega_*$ from Theorem \ref{theorem existence in polar coordidate}. 
\begin{theorem}\label{theorem stability in polar coordinate}
    For all $\varepsilon>0$, there exists $\delta>0$ so that for all $\varphi\in C^{2}_\mathrm{loc}([R,\infty))$ with 
    \begin{equation}
    \sup_r(|r^{-2}\varphi|+|r^{-1}\varphi_r|+|\varphi_{rr}|)<\delta, \quad \varphi_r(R)=0,
    \label{e:reg}
    \end{equation}
    we have that the solution $\Phi(t,r)$ with initial condition $\phi_*(r)+\varphi(r)$ to \eqref{e:pde} satisfies 
    \[
      \|\Phi(t,\cdot)\|_{C^0}<\varepsilon \quad \text{for all } t>0.
    \]
\end{theorem}
We only outline a proof here and relegate a more thorough discussion of regularity and actual asymptotics to forthcoming work. The proof relies on two ingredients:
\begin{enumerate}
    \item local well-posedness and regularity: there exists a unique global solution to \eqref{e:pde} for initial data of the form specified in \eqref{e:reg};
    \item solutions to \eqref{e:pde} obey a comparison principle with sub- and super-solutions giving a priori bounds on $\Phi$ and $\Phi_r$. 
\end{enumerate}
From these two observations, one concludes that the solution with initial condition $\phi_*+\varphi$ is pointwise bounded by the translated solutions $\phi_*(r)\pm\varepsilon$ for all times, thus establishing the claim. 

We first discuss local well-posedness. We write \eqref{e:pde} in the general form 
\begin{equation}\label{e:pdea}
\Phi_t=a(\Phi_r,r)\Phi_{rr}+b(\Phi_r,r)-\omega_*, \qquad \Phi_r(R)=0,
\end{equation}
which possesses the family of equilibrium solutions $\phi_*(r)+\gamma$, $\gamma\in\R$. In fact, the family of solutions with different $R$- and $\omega$-values, $\phi_*(r;\omega)$ provide a yet larger family of solutions of the equation, with $\phi_{*,r}(r;\omega_1)>\phi_{*,r}(r;\omega_2)$ when $\omega_1>\omega_2$, and satisfying inhomogeneous Neumann boundary conditions,  $\phi_{*,r}(R;\omega)(\omega-\omega_*)>0$. Expansions at infinity readily give $\phi_*(r;\omega)=\omega r + \lambda_1 \log r+\rmO(1)$, where we set $D=V=1$ and where $\lambda_1$ is from Proposition \ref{p:inf}, and $\phi_{*,r}(r;\omega)=\omega+\rmO(1/r)$. 

Inserting the ansatz $\Phi=\phi_*(r;\omega_*)+\varphi$ into \eqref{e:pdea} gives an equation of the form 
\begin{equation}\label{e:pdepert}
\varphi_t=\tilde{a}(\varphi_r,r)\varphi_{rr} + \tilde{b}(\varphi_r,r)\varphi_r,
\end{equation}
where the coefficients have expansions
\begin{align*}
\tilde{a}(\varphi_r,r)&=\frac{1}{r^2} (1+\tilde{a}_1(\varphi_r,r)),\\
\tilde{b}(\varphi_r,r)&=(-1 + \frac{1}{r^2} \tilde{b}_1(\varphi_r,r)).
\end{align*}
The equation \eqref{e:pdepert} is not obviously well-posed since the drift term $-\varphi_r $ is not bounded relative to the degenerate diffusion term $\frac{1}{r^2}\varphi_{rr}$, which poses potential difficulties when treating quasilinear terms of the form $\varphi_r\varphi_{rr}$ perturbatively. 
We therefore eliminate the drift term by means of a time dependent coordinate change in the independent variable, setting
\begin{equation}
    \rho=\left\{\begin{array}{ll} r+\log r -t,& r\geq e^t+1,\\
    r,& r\leq e^t,
    \end{array}\right.
\end{equation}
smoothly interpolated on $r\in(e^t,e^t+1)$. 
One finds $\rho_r=1+\rmO(1/r)$, for all $r$,  uniformly in $t$, and $\rho_t=-1$ on $r\geq e^t$, so that in the new coordinates with $\psi(\rho(r,t),t)=\varphi(r,t)$, 
\begin{equation}\label{e:pderho}
\psi_t=\alpha(\psi_\rho,\rho,t)\psi_{\rho\rho} + \beta(\psi_\rho,\rho,t)\psi_\rho,
\end{equation}
with 
\begin{align*}
\alpha(\psi_\rho,\rho,t)&=\frac{1}{\rho^2} (1+\alpha_1(\psi_\rho,\rho,t)),\\
\beta(\psi_\rho,\rho,t)&= \frac{1}{\rho} \beta_1(\psi_\rho,\rho,t),
\end{align*}
and error terms $\alpha_1=\rmO(1/\rho,\Psi_\rho)$ and $\beta_1$ uniformly bounded on compact sets for any finite time interval $t\in[0,T]$.

The principal part of \eqref{e:pderho} is given through the linear equation
\[
\psi_t=\frac{1}{\rho^2}\psi_{\rho\rho},\qquad \rho\geq R, 
\]
with Neumann boundary condition. The linear operator $\mathcal{A}=\frac{1}{\rho^2}\partial_{\rho\rho},$ with Neumann boundary conditions  generates an analytic semigroup with maximal regularity properties \cite{lunardi} on $X=C^0_{-2}([R,\infty))$, where $C^0_{m}$ is the subset of $C^0_\mathrm{loc}$ with bounded norm
\[
\|u\|_{C^0_{-m}}=\sup_{r\geq R}|r^{m}u(r)|.
\]
The  domain is given through
\[
\mathcal{D}(\mathcal{A})=\{u\in C^0_{-2},\quad u_r\in C^0_{-1}, \quad u_{rr}\in C^0\}.
\]
These facts can be readily inferred by setting $v(r)=r^2 u(r)$, inspecting the conjugate equation for $v$, and introducing the quadratic variable $r=s^2$, which altogether identify $\mathcal{A}$ as conjugate to $\partial_{ss}$ up to lower-order terms. 

One now finds that the equation is well-posed as a quasilinear equation on $\mathcal{D}(\mathcal{A})$
%
%
with smooth solution on a small time interval $t\in[0,\delta]$. Using standard techniques, we extend the solution then to a maximal time interval of existence $t\in[0,T_+)$.

We claim that $T_+=\infty$.  For this, note that the solution is smooth on $t\in(0,T_+)$ and the derivative $\psi_\rho$ solves a parabolic equation. Using the sub-and super-solutions $\phi_{*,r}(r;\omega_\pm)-\phi_{*,r}(r;\omega_*)$ with $\omega_+>\omega_*>\omega_-$, transformed to $\rho$-coordinates, we readily conclude that $\psi_\rho$ is uniformly bounded on $r\geq R, t\in[0,T_*)$. Inserting this result into \eqref{e:pderho}, we find a linear parabolic equation with bounded continuous coefficients for $\psi$ which shows that $\psi$ has a limit at $t=T_*$ and can thus be extended. This shows that in fact $T_*=+\infty$ and our solution is global. 
Comparison with constants $\psi\equiv \pm\varepsilon$ gives the desired stability.

\section{Discussion}

We established existence, uniqueness, and stability of anchored  spiral waves in a driven curvature flow approximation. Anchoring in our context is reflected in requiring that the curve is perpendicular to the circular boundary of a core region. Existence relies on a somewhat explicit phase plane analysis and a shooting argument. The existence argument can easily be generalized to different boundary conditions, requiring for instance that the curve meet the core region at a fixed angle, possibly different from $\pi/2$. In our shooting approach, this leads to a boundary condition $\ell=\alpha_* \tan\beta$. Invariant regions can then be modified in a somewhat straightforward fashion as follows. For negative $\beta$, one modifies the boundary of $G_\omega$ replacing the horizontal axis $\ell=0$  by $\ell=\alpha_* \tan\beta$, and the upper boundary by a sufficiently steep line with large intercept. For positive $\beta$, one needs to enlarge $B_\omega$ to include a region below part of the line $\ell=\alpha_* \tan\beta$, while retaining the upper boundary of $G_\omega$. Interestingly, the asymptotics in \eqref{theorem existence in polar coordidate} change: for $\beta>0$, the boundary layer is fully developed and $\sigma_0$ is replaced by $\sigma_1$, obtained from the first zero of the Airy function, $\mathrm{Ai}(-\sigma_1)=0,\ \mathrm{Ai}(-\sigma_1)>0$ for $\sigma<\sigma_1$, $\sigma_1=-2.338107410\ldots$; for $\beta>0$, there is no boundary layer, $\omega_*$ has a smooth expansion in $\alpha_*$,   $\omega_*=V\alpha_*+\rmO(\alpha_*^2)$. Geometrically, the difference is a non-monotonicity of the angle in the radius when $\beta<0$, which induces the boundary layer. It would be interesting to study more general, possibly dynamic, boundary  conditions, preserving well-posedness of the PDE, that could lead to saddle-nodes or even Hopf bifurcations and thereby mimic retracting wave or meandering bifurcations.  Similarly, a description including curves that are not graphs over the radial variable might shed light on global dynamics.

Stability relies on a comparison principle in an equation of the rough form 
\begin{equation}\label{e:lin}
\Phi_t=\frac{1}{r^2}\Phi_{rr} - \Phi_r, \quad r\geq R, \qquad \Phi_r=0, \quad r=R.
\end{equation}
Interestingly, it appears that diffusion in this equation is not strong enough to induce decay of localized perturbations. Indeed, heuristically arguing that a small localized perturbation travels with constant speed $1$ towards $r=\infty$, diffusion acts locally on this disturbance with diffusion coefficient $\frac{1}{(t+1)^2}$. In the diffusion equation $\Phi_t=\frac{1}{(1+t)^{2}}\Phi_{rr}$, however, one readily substitutes a new time variable $\tau=1-\frac{1}{1+t}$, to obtain $\tilde{\Phi}_\tau=\tilde{\Phi}_{rr}$, so that  $\Phi_\infty(r)=\tilde{\Phi}(\tau=1,r)$ gives the asymptotic shape, 
\[
\Phi(t,r)\sim \Phi_\infty(r-t), \qquad \text{as } t\to\infty. 
\]
This lack of ``healing'' of an interface has been well studied in the context of reaction-diffusion systems with radial, outward-propagating interfaces; see for instance \cite{roussier1,roussier2}. Our approach circumvents this difficulty by providing rough bounds using comparison functions, and not attempting to capture asymptotics. 

We notice however that the simple curvature model neglects the effect of interaction between spiral arms, which may induce additional damping. In fact, the essential spectrum of the linear operator in \eqref{e:lin} contains the entire imaginary axis, while the essential spectra of spiral waves typically contain parabolic arcs $\Gamma_0\{\lambda=\rmi k - d k^2,k\sim 0\}$ as well as vertical translates $\Gamma_0+\rmi\omega_*\Z$ with constant $d$ induced by the diffusive interaction of wave trains in the radial direction; see \cite{MR4580296}. Neglecting this interaction in our curvature approximation therefore leads to $d=0$, spectrum on the entire imaginary axis, and a lack of decay of localized perturbations. 

We hope that the study here can contribute towards a better understanding of the damping properties in spiral wave dynamics, and ultimately help extend Claudia Wulff's insightful early work on spiral meanders \cite{wulffthesis,wulffthesispub} towards Archimedean spiral waves. 

%

\begin{thebibliography}{10}

\bibitem{PhysRevLett.72.164}
D.~Barkley.
\newblock Euclidean symmetry and the dynamics of rotating spiral waves.
\newblock {\em Phys. Rev. Lett.}, 72:164--167, Jan 1994.

\bibitem{Belmonte1997}
A.~L. {Belmonte}, Q.~{Ouyang}, and J.-M. {Flesselles}.
\newblock {Experimental Survey of Spiral Dynamics in the Belousov-Zhabotinsky
  Reaction}.
\newblock {\em Journal de Physique II}, 7(10):1425--1468, Oct. 1997.

\bibitem{BERNOFF1991125}
A.~J. Bernoff.
\newblock Spiral wave solutions for reaction-diffusion equations in a fast
  reaction/slow diffusion limit.
\newblock {\em Physica D: Nonlinear Phenomena}, 53(1):125--150, 1991.

\bibitem{desai_kapral_2009}
R.~C. Desai and R.~Kapral.
\newblock {\em Dynamics of Self-Organized and Self-Assembled Structures}.
\newblock Cambridge University Press, 2009.

\bibitem{MR524817}
N.~Fenichel.
\newblock Geometric singular perturbation theory for ordinary differential
  equations.
\newblock {\em J. Differential Equations}, 31(1):53--98, 1979.

\bibitem{MR2341216}
B.~Fiedler, M.~Georgi, and N.~Jangle.
\newblock Spiral wave dynamics: reaction and diffusion versus kinematics.
\newblock In {\em Analysis and control of complex nonlinear processes in
  physics, chemistry and biology}, volume~5 of {\em World Sci. Lect. Notes
  Complex Syst.}, pages 69--114. World Sci. Publ., Hackensack, NJ, 2007.

\bibitem{MR2094384}
B.~Fiedler, J.-S. Guo, and J.-C. Tsai.
\newblock Multiplicity of rotating spirals under curvature flows with normal
  tip motion.
\newblock {\em J. Differential Equations}, 205(1):211--228, 2004.

\bibitem{MR2187763}
B.~Fiedler, J.-S. Guo, and J.-C. Tsai.
\newblock Rotating spirals of curvature flows: a center manifold approach.
\newblock {\em Ann. Mat. Pura Appl. (4)}, 185:S259--S291, 2006.

\bibitem{fssw}
B.~Fiedler, B.~Sandstede, A.~Scheel, and C.~Wulff.
\newblock Bifurcation from relative equilibria of noncompact group actions:
  skew products, meanders, and drifts.
\newblock {\em Doc. Math.}, 1:No. 20, 479--505, 1996.

\bibitem{GLM}
M.~Golubitsky, V.~G. LeBlanc, and I.~Melbourne.
\newblock Meandering of the spiral tip: an alternative approach.
\newblock {\em J. Nonlinear Sci.}, 7(6):557--586, 1997.

\bibitem{MR0588502}
J.~M. Greenberg.
\newblock Spiral waves for {$\lambda -\omega $} systems.
\newblock {\em SIAM J. Appl. Math.}, 39(2):301--309, 1980.

\bibitem{MR0639765}
J.~M. Greenberg.
\newblock Spiral waves for {$\lambda -\omega $} systems. {II}.
\newblock {\em Adv. in Appl. Math.}, 2(4):450--454, 1981.

\bibitem{MR0665385}
P.~S. Hagan.
\newblock Spiral waves in reaction-diffusion equations.
\newblock {\em SIAM J. Appl. Math.}, 42(4):762--786, 1982.

\bibitem{MR1629103}
R.~Ikota, N.~Ishimura, and T.~Yamaguchi.
\newblock On the structure of steady solutions for the kinematic model of
  spiral waves in excitable media.
\newblock {\em Japan J. Indust. Appl. Math.}, 15(2):317--330, 1998.

\bibitem{KEENER1986307}
J.~P. Keener and J.~J. Tyson.
\newblock Spiral waves in the belousov-zhabotinskii reaction.
\newblock {\em Physica D: Nonlinear Phenomena}, 21(2):307--324, 1986.

\bibitem{MR0359550}
N.~Kopell and L.~N. Howard.
\newblock Plane wave solutions to reaction-diffusion equations.
\newblock {\em Studies in Appl. Math.}, 52:291--328, 1973.

\bibitem{MR0639764}
N.~Kopell and L.~N. Howard.
\newblock Target pattern and spiral solutions to reaction-diffusion equations
  with more than one space dimension.
\newblock {\em Adv. in Appl. Math.}, 2(4):417--449, 1981.

\bibitem{MR1857972}
M.~Krupa and P.~Szmolyan.
\newblock Extending geometric singular perturbation theory to nonhyperbolic
  points---fold and canard points in two dimensions.
\newblock {\em SIAM J. Math. Anal.}, 33(2):286--314, 2001.

\bibitem{lunardi}
A.~Lunardi.
\newblock {\em Analytic semigroups and optimal regularity in parabolic problems}
\newblock Birkh\"{a}user/Springer Basel AG, Basel, 1995.

\bibitem{Barkley_cookbook}
D.~Margerit and D.~Barkley.
\newblock {Cookbook asymptotics for spiral and scroll waves in excitable
  media}.
\newblock {\em Chaos: An Interdisciplinary Journal of Nonlinear Science},
  12(3):636--649, 08 2002.

\bibitem{MIKHAILOV1991379}
A.~Mikhailov and V.~Zykov.
\newblock Kinematical theory of spiral waves in excitable media: Comparison
  with numerical simulations.
\newblock {\em Physica D: Nonlinear Phenomena}, 52(2):379--397, 1991.

\bibitem{MR1257848}
A.~S. Mikha\u{\i}lov, V.~A. Davydov, and V.~S. Zykov.
\newblock Complex dynamics of spiral waves and motion of curves.
\newblock {\em Phys. D}, 70(1-2):1--39, 1994.

\bibitem{MR4656059}
L.~Pismen.
\newblock {\em Patterns and {I}nterfaces in {D}issipative {D}ynamics}.
\newblock Springer Series in Synergetics. Springer, Cham, nd edition, 2023.
\newblock Revised and Extended, Now also Covering Patterns of Active Matter.

\bibitem{roussier2}
J.-M. Roquejoffre and V.~Roussier-Michon.
\newblock Sharp large time behaviour in {$N$}-dimensional reaction-diffusion
  equations of bistable type.
\newblock {\em J. Differential Equations}, 339:134--151, 2022.

\bibitem{roussier1}
V.~Roussier.
\newblock Stability of radially symmetric travelling waves in
  reaction-diffusion equations.
\newblock {\em Ann. Inst. H. Poincar\'{e} C Anal. Non Lin\'{e}aire},
  21(3):341--379, 2004.

\bibitem{MR4580296}
B.~Sandstede and A.~Scheel.
\newblock Spiral waves: linear and nonlinear theory.
\newblock {\em Mem. Amer. Math. Soc.}, 285(1413):v+126, 2023.

\bibitem{ssw2}
B.~Sandstede, A.~Scheel, and C.~Wulff.
\newblock Dynamics of spiral waves on unbounded domains using center-manifold
  reductions.
\newblock {\em J. Differential Equations}, 141(1):122--149, 1997.

\bibitem{ssw3}
B.~Sandstede, A.~Scheel, and C.~Wulff.
\newblock Bifurcations and dynamics of spiral waves.
\newblock {\em J. Nonlinear Sci.}, 9(4):439--478, 1999.

\bibitem{ssw1}
B.~Sandstede, A.~Scheel, and C.~Wulff.
\newblock Dynamical behavior of patterns with {E}uclidean symmetry.
\newblock In {\em Pattern formation in continuous and coupled systems
  ({M}inneapolis, {MN}, 1998)}, volume 115 of {\em IMA Vol. Math. Appl.}, pages
  249--264. Springer, New York, 1999.

\bibitem{MR1638058}
A.~Scheel.
\newblock Bifurcation to spiral waves in reaction-diffusion systems.
\newblock {\em SIAM J. Math. Anal.}, 29(6):1399--1418, 1998.

\bibitem{TYSON1988327}
J.~J. Tyson and J.~P. Keener.
\newblock Singular perturbation theory of traveling waves in excitable media (a
  review).
\newblock {\em Physica D: Nonlinear Phenomena}, 32(3):327--361, 1988.

\bibitem{MR0025140}
N.~Wiener and A.~Rosenblueth.
\newblock The mathematical formulation of the problem of conduction of impulses
  in a network of connected excitable elements, specifically in cardiac muscle.
\newblock {\em Arch. Inst. Cardiol. M\'{e}xico}, 16:205--265, 1946.

\bibitem{MR0572965}
A.~T. Winfree.
\newblock {\em The geometry of biological time}, volume~8 of {\em
  Biomathematics}.
\newblock Springer-Verlag, Berlin-New York, 1980.

\bibitem{wulffthesis}
C.~Wulff.
\newblock {\em Theory of meandering and drifting spiral waves in reaction-
  diﬀusion systems.}
\newblock PhD thesis, FU Berlin, 1996.

\bibitem{wulffthesispub}
C.~Wulff.
\newblock Transitions from relative equilibria to relative periodic orbits.
\newblock {\em Doc. Math.}, 5:227--274, 2000.

\bibitem{YAMADA1993153}
H.~Yamada and K.~Nozaki.
\newblock Dynamics of wave fronts in excitable media.
\newblock {\em Physica D: Nonlinear Phenomena}, 64(1):153--162, 1993.

\end{thebibliography}

\end{document}